\documentclass[reqno]{amsart}

\usepackage{algorithms}

\usepackage{amsmath}
\usepackage{amssymb}
\usepackage{amsfonts}
\usepackage{amsthm}
\usepackage[foot]{amsaddr}

\usepackage{mathtools}
\mathtoolsset{%
}

\usepackage[utf8]{inputenc}
\usepackage[T1]{fontenc}

\usepackage{dsfont}
\usepackage[%
cal=cm,
]
{mathalfa}

\usepackage[dvipsnames,svgnames]{xcolor}
\colorlet{MyBlue}{DodgerBlue!75!Black}
\colorlet{MyGreen}{DarkGreen!85!Black}

\usepackage{subfigure}
\usepackage{tikz}

\usepackage{acronym}
\usepackage{latexsym}
\usepackage{paralist}
\usepackage{wasysym}
\usepackage{xspace}

\usepackage[numbers,sort&compress]{natbib}

\usepackage{hyperref}
\hypersetup{
colorlinks=true,
linktocpage=true,
pdfstartview=FitH,
breaklinks=true,
pdfpagemode=UseNone,
pageanchor=true,
pdfpagemode=UseOutlines,
plainpages=false,
bookmarksnumbered,
bookmarksopen=false,
bookmarksopenlevel=1,
hypertexnames=true,
pdfhighlight=/O,
urlcolor=Maroon,linkcolor=MyBlue!60!black,citecolor=DarkGreen!70!black, % <--- for screen
pdftitle={},
pdfauthor={},
pdfsubject={},
pdfkeywords={},
pdfcreator={pdfLaTeX},
pdfproducer={LaTeX with hyperref}
}

\def\EMAIL#1{\href{mailto:#1}{#1}}

\newcommand{\eps}{\varepsilon}
\newcommand{\from}{\colon}

\newcommand{\pd}{\partial}

\newcommand{\C}{\mathbb{C}}
\newcommand{\R}{\mathbb{R}}

\DeclareMathOperator*{\argmax}{arg\,max}

\DeclareMathOperator*{\intersect}{\bigcap}

\DeclareMathOperator{\bigoh}{\mathcal O}

\DeclareMathOperator{\ex}{\mathbb{E}}

\DeclareMathOperator{\prob}{\mathbb{P}}

\DeclareMathOperator{\supp}{supp}

\providecommand\given{} % provides an empty command for the delimiters below

\DeclarePairedDelimiter{\abs}{\lvert}{\rvert}
\DeclarePairedDelimiter{\norm}{\lVert}{\rVert}
\newcommand{\dnorm}[1]{\norm{#1}_{\infty}}

\DeclarePairedDelimiterX{\braket}[2]{\langle}{\rangle}{#1\mathopen{}\delimsize\vert\mathopen{}#2}

\DeclarePairedDelimiterX{\inner}[2]{\langle}{\rangle}{#1,#2}
\DeclarePairedDelimiterX{\setdef}[2]{\{}{\}}{#1:#2}

\DeclarePairedDelimiterXPP\exclude[1]{\mathopen{}\setminus}{\{}{\}}{}{#1}

\DeclarePairedDelimiterXPP{\probof}[1]{\prob}{(}{)}{}{%
\renewcommand\given{\nonscript\,\delimsize\vert\,\nonscript\mathopen{}}
#1}

\DeclarePairedDelimiterXPP{\exof}[1]{\ex}{[}{]}{}{%
\renewcommand\given{\nonscript\,\delimsize\vert\,\nonscript\mathopen{}}
#1}

\DeclarePairedDelimiterXPP{\modexof}[1]{\ex'}{[}{]}{}{%
\renewcommand\given{\nonscript\,\delimsize\vert\nonscript\,\mathopen{}}
#1}

\newcommand{\txs}{\textstyle}
\newcommand{\textpar}[1]{\textup(#1\textup)}

\newcommand{\insum}{\sum\nolimits}

\newcommand{\as}{\textup(a.s.\textup)\xspace}

\usepackage[textwidth=30mm]{todonotes}
\usepackage{soul}
\setstcolor{red}
\sethlcolor{SkyBlue}

\theoremstyle{plain}
\newtheorem{theorem}{Theorem}
\newtheorem{corollary}[theorem]{Corollary}
\newtheorem*{corollary*}{Corollary}

\newtheorem{proposition}[theorem]{Proposition}

\theoremstyle{definition}

\newtheorem*{definition*}{Definition}

\theoremstyle{remark}
\newtheorem{remark}{Remark}
\newtheorem*{remark*}{Remark}

\newcommand{\bvec}{e}

\newcommand{\fench}{F}

\DeclareMathOperator{\logit}{\Lambda}

\newcommand{\play}{i}

\newcommand{\nPlayers}{N}
\newcommand{\players}{\mathcal{\nPlayers}}

\newcommand{\pure}{s}
\newcommand{\purealt}{\pure'}
\newcommand{\nPures}{S}
\newcommand{\pures}{\mathcal{\nPures}}

\newcommand{\actions}{\mathcal{X}}
\newcommand{\strats}{\actions}

\newcommand{\pay}{u}
\newcommand{\payv}{v}

\newcommand{\eq}{x^{\ast}}

\newcommand{\peq}{\pure^{\ast}}

\newcommand{\argdot}{\,\cdot\,}
\newcommand{\dkl}{D_{\textup{KL}}}
\newcommand{\filter}{\mathcal{F}}

\DeclareMathOperator{\simplex}{\Delta}
\newcommand{\intsimplex}{\simplex^{\!\circ}}

\newcommand{\step}{\gamma}

\newcommand{\noise}{\xi}
\newcommand{\snoise}{\psi}
\newcommand{\noisedev}{\sigma}
\newcommand{\noisevar}{\noisedev^{2}}
\newcommand{\vbound}{L}

\begin{document}

%*************************************************************
%*****    FRONT MATTER AND METADATA
%*************************************************************

%----------------------------------------------------------------------
%%% TITLE & AUTHORS
%----------------------------------------------------------------------
\title
[Exponentially fast convergence to (strict) equilibrium via hedging]
{Exponentially fast convergence to\\(strict) equilibrium via hedging}

\author[J.~Cohen]{Johanne Cohen$^{\ast}$}
\address{$^{\ast}$ LRI-CNRS, Université Paris-Sud,Université Paris-Saclay, France}
\email{\EMAIL{johanne.cohen@lri.fr}}
%\urladdr{\URL{http://www.lri.fr/\~jcohen}}

\author[A.~Héliou]{Amélie Héliou$^{\ddag}$}
\address{$\ddag$ LIX, Ecole Polytechnique,Université Paris-Saclay, France}
\email{\EMAIL{amelie.heliou@polytechnique.edu}}
%\urladdr{\URL{http://www.lix.polytechnique.fr/Labo/Amelie.Heliou/}}

\author[P.~Mertikopoulos]{Panayotis Mertikopoulos$^{\S}$}
\address{$^{\S}$ LIG-CNRS, Univ. Grenoble-Alpes, Grenoble, France}
\email{\EMAIL{panayotis.mertikopoulos@imag.fr}}
%\urladdr{\URL{http://mescal.imag.fr/membres/panayotis.mertikopoulos}}

%\author{Johanne Cohen\inst{1} \and Am\'elie H\'eliou\inst{2} \and  Panayotis Mertikopoulos\inst{3}}
%%
%%\authorrunning{Johanne Cohen et al.} % abbreviated author list (for running head)
%%
%%%%% list of authors for the TOC (use if author list has to be modified)
%\tocauthor{Johanne Cohen, Am\'elie H\'eliou, and  Panayotis Mertikopoulos}
%%
%\institute{LRI-CNRS, Universit\'{e} de Paris-Sud,Universit\'{e} Paris-Saclay, France,\\
%\email{johanne.cohen@lri.fr},\\ WWW home page:
%\texttt{http://www.lri.fr/\homedir jcohen}
%\and
%LIX, Ecole Polytechnique,Universit\'{e} Paris-Saclay, France,\\
%\email{amelie.heliou@polytechnique.edu},\\  
%WWW home page:
%\texttt{http://www.lix.polytechnique.fr/Labo/Amelie.Heliou/}
%\and LIG-CNRS, Grenoble, France \\ \email{panayotis.mertikopoulos@imag.fr}, \\
%WWW home page:
%\texttt{http://mescal.imag.fr/membres/panayotis.mertikopoulos}
%}

\maketitle

%----------------------------------------------------------------------
%%% ACRONYMS
%----------------------------------------------------------------------
\newcommand{\acdef}[1]{\textit{\acl{#1}} \textup{(\acs{#1})}\acused{#1}}
\newcommand{\acdefp}[1]{\emph{\aclp{#1}} \textup(\acsp{#1}\textup)\acused{#1}}

\newacro{iid}[i.i.d.]{independent and identically distributed}
\newacro{KKT}{Karush\textendash Kuhn\textendash Tucker}
\newacro{KL}{Kullback\textendash Leibler}
\newacro{EW}{exponential weights}
\newacro{NE}{Nash equilibrium}
\newacroplural{NE}[NE]{Nash equilibria}
\newacro{GESS}{globally evolutionarily stable state}
\newacro{MSE}{mean squared error}

%----------------------------------------------------------------------
%%% ABSTRACT
%----------------------------------------------------------------------
\begin{abstract}
%----------------------------------------------------------------------
%%% ABSTRACT
%----------------------------------------------------------------------
% !TEX root = ./WINE.tex

Motivated by applications to data networks where fast convergence is essential, we analyze the problem of learning in generic $\nPlayers$-person games that admit a \acl{NE} in pure strategies.
Specifically, we consider a scenario where players interact repeatedly and try to learn from past experience by small adjustments based on local \textendash\ and possibly imperfect \textendash\ payoff information.
For concreteness, we focus on the so-called ``hedge'' variant of the \ac{EW} algorithm where players select an action with probability proportional to the exponential of the action's cumulative payoff over time.
When the players have perfect information on their mixed payoffs, the algorithm converges locally to a strict equilibrium and the rate of convergence is exponentially fast \textendash\ of the order of $\bigoh(\exp(-a\sum_{j=1}^{t}\step_{j}))$ where $a>0$ is a constant and $\step_{j}$ is the algorithm's step-size.
In the presence of uncertainty, convergence requires a more conservative step-size policy, but with high probability, the algorithm still achieves an exponential convergence rate.
%We focus on learning equilibria on static finite games in normal form. We consider a repeated game, in which a set of players are in interaction, and make, at each iteration, a decentralized decision on which strategy to play. 
%We assume that players take decision with only a local view of the game.
%We consider two settings.
%In the first one, we assume that players know of the payoff of mixed strategies. We consider a variant of the  HEDGE algorithm  \cite{FS99} with the introduction of a discounting factor $\gamma_{t}$.
%We prove that this algorithm is guaranteed to converge to any strict equilibria $x^*$, if initialized not too far from $x^*$ and the discounting factor is chosen  in a simple suitable way.
%In the second settings, we prove a similar result assuming that players have only a knowledge of payoffs of  pure strategies.
%In both cases, we provide some bounds on the rate of convergence of the algorithm: the rate of convergence  is proved to be exponential when  the discounting factor is constant. We also prove that in both settings, the results still hold even when players have only noisy observations.
%Our proofs are based heavily on the Kullback-Leibler divergence, martingale arguments and minor gradient arguments.

\end{abstract}
\acresetall

%----------------------------------------------------------------------
%%% KEYWORDS
%----------------------------------------------------------------------
%\subjclass[2010]{Primary ; Secondary .}
%\keywords{}

%*************************************************************
%*****    BODY TEXT
%*************************************************************

%----------------------------------------------------------------------
%%% INTRODUCTION
%----------------------------------------------------------------------
\section{Introduction}
\label{sec:introduction}
%----------------------------------------------------------------------
%%% INTRODUCTION
%----------------------------------------------------------------------
% !TEX root = ./WINE.tex

This paper is a contribution to the following questions:
if the players of a repeated game update their strategies in an individually rational way (for instance, following an algorithm that leads to no regret), do the players' mixed strategies themselves converge to a \acl{NE} of the one-shot game?
If so, what is the resulting rate of convergence and how is it affected by imperfections in the information available to the players?

These questions are largely motivated by the extremely successful applications of game theory to data networks \cite{NRTV07} and wireless communications \cite{LT10} where fast convergence to a stable state is essential.
%for the system's efficient operation.
In this broad context, players naturally have a very localized view of their environment, typically limited to an estimate of the payoff of their strategies that is often subject to random \textendash\ and possibly unbounded \textendash\ errors and noise.
Thus, to improve their individual payoffs as the game is repeated, we assume that players try to learn from past experiences by employing a suitable learning algorithm that induces small adjustments at each stage.

One of the most widely used learning algorithms of this kind is the so-called \acdef{EW} algorithm which is known to lead to ``no regret'' \textendash\ i.e. the players' average payoff under \ac{EW} is asymptotically the same as that of the best fixed strategy in hindsight \cite{Vov90,LW94}.
More precisely, we focus on the ``hedge'' variant of \ac{EW} \cite{FS99} where the probability of choosing an action is proportional to the exponential of its cumulative payoff \textendash\ hence, better-performing actions are employed exponentially more often.
However, to account for the fact that players may not have access to perfect payoff observations, we introduce a variable step-size parameter which can be used to control the weight with which new observations enter the algorithm, thus reducing the adverse effects of uncertainty.

Instead of restricting our attention to a specific class of games (such as zero-sum or potential ones), we focus throughout on generic $\nPlayers$-player games that admit a \acl{NE} in pure strategies.
Our first result is that, if players have perfect mixed payoff observations, hedging with a variable step-size $\step_{t}>0$ converges locally to a strict \acl{NE} at a rate of $\bigoh(e^{-a\sum_{j=1}^{t}\step_{j}})$, i.e. \emph{exponentially fast} if $\step_{t} = \Omega(1/t^{\beta})$ for some $\beta<1$.
Otherwise, if players only have access to pure payoff observations that are subject to estimation errors, the same convergence rate holds with high probability, but the algorithm's step-size must satisfy the summability requirement $\sum_{t=1}^{\infty} \step_{t}^{2} < \sum_{t=1}^{\infty} \step_{t} = \infty$.
This restriction is needed in order to control the aggregate variance of the noise so, unsurprisingly, it limits the achievable convergence rates.
Nevertheless, with high probability, the algorithm still achieves an $\bigoh(e^{-at^{1-\beta}})$ exponential rate for step-size sequences of the form $\step_{t}\sim 1/t^{\beta}$, $\beta\in(1/2,1)$, despite the noise.
Finally, we show that the above results hold globally if the game admits a unique, globally strict equilibrium.

\paragraph{Related Work.}

Algorithms and dynamics for learning in games have received considerable attention over the last few decades.
Such procedures can be divided into two broad categories, depending on whether they evolve in continuous or discrete time: the former includes the numerous dynamics for learning and evolution (see \cite{San10} for a survey), whereas the latter focuses on learning algorithms for infinitely iterated games (such as fictitious play and its variants).
In this paper, we focus exclusively on discrete-time algorithms.

In this framework, it is natural to consider agents who learn from their experience by small adjustments in their behavior based on local \textendash\ and possibly imperfect \textendash\ information.
%In this case, the dynamics governing the agents' behavior are typically stochastic and distributed, owing to the fact that agents adapt their mixed strategies based on local \textendash\ and possibly imperfect \textendash\ information.
%local knowledge only, and  using mixed strategies.
%Several such adjustment dynamics have been considered recently in the literature.
Several such approaches in the literature can be viewed as \emph{decentralized no-regret dynamics} \textendash\ for example the multiplicative/\acl{EW} algorithm and its variants \cite{Vov90,FS99,LW94},
%Mirror Descent \cite{yudin1983problem},
%and
Follow the Regularized/Perturbed Leader \cite{kalai2005efficient},
etc.
Indeed, regret bounds can be used to guarantee that each player's utility approaches long-term optimality in adversarial environments, a natural first step towards long-term rational behavior.
For example, it has been shown in \cite{BHL08,roughgarden2015intrinsic} that the sum of utilities approaches an approximate optimum, and there is convergence of time averages towards an equilibrium in two-player zero-sum games \cite{BM07a,FV97,FS99}.
In all these examples, the players' average regret vanishes at the worst-case rate of $\bigoh(1/\sqrt{T})$ where $T$ denotes the play horizon.
This convergence rate was recently improved by Syrgkanis \emph{et al.} \cite{NIPS2015} for a wide class of $\nPlayers$-player normal form games using a natural class of regularized learning algorithms.
However, the convergence results established in \cite{NIPS2015} concerned the set of coarse correlated equilibria which may contain highly non-rationalizable (correlated) strategies that assign positive weight \emph{only} on strictly dominated strategies \cite{VZ13}.

In this paper, we aim to provide a more refined analysis for generic $\nPlayers$-player games that admit a \acl{NE} in pure strategies.
In particular, we do not derive convergence bounds for unilateral worst-case rationality criteria (such as the minimization of the players' external regret), but we focus squarely on the convergence of the players' mixed strategies under hedging.
% is not on some external quantities such as the average regret, but on the actual convergence of the dynamics.
%For a variant of the HEDGE dynamics, we prove that when started close to some strict equilibria a convergence occurs, by studying the associated dynamical system. Furthermore, we establish the rate of convergence of the process.  
To that end, we show that HEDGE converges locally to strict equilibria and the rate of said convergence is exponentially fast \textendash\ even in the presence of feedback imperfections of \emph{arbitrary} magnitude.

%We now go to a comparison to some related papers.
The HEDGE algorithm was recently studied by Kleinberg \emph{et al.} \cite{kleinberg2011load} who proved that, in a specific class of load balancing games, the dynamics' long-term limit is exponentially better than the worst correlated equilibrium and almost as good as that of the worst Nash.
%Their proof focuses on particular load balancing games and it is based on the so-called \acdef{KL} divergence.
Krichene \emph{et al.} \cite{krichene2014learning} extended this result to congestion games, and proved that a discounted variant of the HEDGE algorithm converges to the set of Nash equilibria in the sense of Cesàaro means (time averages), while strong convergence can be guaranteed with some additional conditions.
%It is proved that (strong) convergence can be guaranteed with some additional conditions.
%In the current paper, we do not restrict to congestion games, and consider general games.
Coucheney \emph{et al.} \cite{CGM15} also showed that a ``penalty-regulated'' variant of the HEDGE algorithm with bandit feedback converges to $\eps$-equilibrium in congestion games, but their techniques do not extend to actual Nash equilibria.

%In many prediction problems the agent, after taking an action, is able to estimate his reward but does not have access to an other information. Such prediction problems have been known as multi-armed bandit problems \cite{Rob52}. At each time step, a unit resource is allocated to an action and some observable payoff is obtained. The goal is to maximize the total payoff obtained in a sequence of allocations.

%Many situations involve repeatedly making decisions in an uncertain environment. In the literature, some algorithms that minimize an even stronger form of regret, known as internal or swap regret have been considered to compute some equilibrium. Internal regret is important because when all players in a game minimize this stronger type of regret, the empirical distribution of play is known to converge to correlated equilibrium \cite{HMC00}. For several classical games, several works show that if all participants minimize their own external regret, then overall traffic is guaranteed to converge to an approximate Nash Equilibrium.

%The speed at which the game outcome converges to this approximately optimal welfare is governed by individual players? regret bounds. There is a large number of simple regret minimization algorithms (multiplicative weights \cite{LW94}, mirror decent, follow the regularized leader \cite{FV97}, see e.g. \cite{Haz16}) guaranteeing that the average regret goes down as O(1/ ?T) with time T, which is tight in adver- sarial settings.

\paragraph{Organization of the paper.} 
Section~\ref{sec:prelims} provides some definitions and preliminaries used in the rest of the paper.
In Section~\ref{sec:HEDGE}, we describe the HEDGE algorithm and the players' feedback and information models.
Section~\ref{sec:results} is split into two parts:
The first has the statement of our results (convergence and rate of convergence), while the second one contains the mathematical apparaturs required to prove these results.
Section \ref{sec:conclusion} concludes while some technical details have been relegated to Appendix \ref{app:auxiliary}.

%----------------------------------------------------------------------
%%% PRELIMINARIES
%----------------------------------------------------------------------
\section{Preliminaries}
\label{sec:prelims}
%----------------------------------------------------------------------
%%% PRELIMINARIES
%----------------------------------------------------------------------
% !TEX root = ./WINE.tex

We begin with some basic definitions from game theory.
Throughout the paper, we focus on games that are played by a (finite) set $\players = \{1,\dotsc,\nPlayers\}$ of $\nPlayers$ \emph{players} (or \emph{agents}).
Each player $\play\in\players$ is assumed to have a finite set of \emph{actions} (or \emph{pure strategies}) $\pures_{\play}$, and the players' preferences for one action over another are represented by each action's \emph{utility} (or \emph{payoff}).
Specifically, as players interact with each other, the individual payoff of each player is given by a function $\pay_{\play}\from \pures\equiv \prod_{\play} \pures_{\play} \to \R$ of all players' actions, and each agent seeks to maximize the utility $\pay_{\play}(\pure_{\play};\pure_{-\play})$ of his chosen action $\pure_{\play}\in\pures_{\play}$ against the action profile $\pure_{-\play}$ of his opponents.%
\footnote{In the above $(\pure_{\play};\pure_{-\play})$ is shorthand for $(\pure_{1},\dotsc,\pure_{\play},\dotsc,\pure_{\nPlayers})$, used here to highlight the action of player $\play$ against that of all other players.}
%Let $\pay_{\play}(\pure_{1},\dotsc,\pure_{\nPlayers})$ denote the payoff of player $\play$ in the pure strategy profile $(\pure_{1},\dotsc,\pure_{\nPlayers})$ . 
%The \emph{maximum absolute payoff} will be denoted as $L = \max_{\play\in\players} \max_{x\in\strats} \pay_{\play}(x)$.
%Obviously $L<\infty$.
%Putting all this together, a finite game in normal form will be writen as a tuple $\game \equiv \gamefull$ with players, pure strategies and payoff functions defined as above.

A player can also use a \emph{mixed strategy} by playing a probability distribution $x_{\play}=(x_{\play\pure_{\play}} )_{\pure_{\play}\in \pures_{\play}} \in \simplex(\pures_{\play})$ over their action set $\pures_{\play}$.
The resulting probability vector $x_{\play}$ is called the \emph{mixed strategy} of the $\play$-th player and the set $\strats_{\play} = \simplex(\pures_{\play})$ is  the corresponding mixed strategy space of player $\play$.
Based on this, we write $\strats = \prod_{\play} \strats_{\play}$ for the game's \emph{strategy space}, i.e. the space of all mixed strategy profiles $x=(x_{\play})_{\play\in\players}$.

In this context (and in a slight abuse of notation), the expected payoff of the $\play$-th player in the mixed strategy profile $x = (x_{1},\dotsc,x_{\nPlayers})$ is
\begin{equation}
\pay_{\play}(x)
	= \sum_{\mathclap{\pure_{1}\in\pures_{1}}} \:\dotsi\: \sum_{\mathclap{\pure_{\nPlayers}\in\pures_{\nPlayers}}} \;
	\pay_{\play}(\pure_{1},\dotsc,\pure_{\nPlayers})\,
	x_{1\pure_{1}} \dotsm x_{\nPlayers\pure_{\nPlayers}},
\end{equation}
Accordingly, if player $\play$ plays the pure strategy $\pure_{\play}$ in $\pures_{\play}$, we will write
\begin{equation}
\payv_{\play\pure_{\play}}(x)
	= \pay_{\play}(\pure_{\play};x_{-\play})
	= \pay_{\play}(x_{1},\dotsc,\pure_{\play},\dotsc, x_{\nPlayers})
\end{equation}
for the payoff corresponding to the pure strategy $\pure_{\play}\in\pures_{\play}$ and $\payv_{\play}(x) = (\payv_{\play\pure_{\play}}(x))_{\pure_{\play}\in\pures_{\play}}$ for the \emph{payoff vector} of player $\play$.
A player's expected payoff can thus be written as
\begin{equation}
\label{eq:braket}
\pay_{\play}(x)
	= \sum_{\pure_{\play}\in\pures_{\play}} x_{\play\pure_{\play}} \payv_{\play\pure_{\play}}(x)
	= \braket{\payv_{\play}(x)}{x_{\play}},
\end{equation}
where $\braket{\payv_{\play}}{x_{\play}}$ denotes the canonical bilinear pairing between $\payv_{\play}$ and $x_{\play}$.

The most widely used solution concept in game theory is that of a \acdef{NE}, i.e. a state $\eq\in\strats$ that is unilaterally stable in the sense that
\begin{equation}
\label{eq:NE}
\tag{NE}
\pay_{\play}(\eq_{\play};\eq_{-\play})
	\geq \pay_{\play}(x_{\play};\eq_{-\play})
	\quad
	\text{for all $x_{\play}\in\strats_{\play}$, $\play\in\players$},
\end{equation}
or, equivalently, writing $\supp(x)$ for the support of $x$:
\begin{equation}
\label{eq:NE-pures}
\payv_{\play\pure_{\play}}(\eq)
	\geq \payv_{\play\pure_{\play}'}(\eq)
	\quad
	\text{for all $\pure_{\play}\in\supp(\eq_{\play})$ and all $\pure_{\play}'\in\pures_{\play}$, $\play\in\players$}.
\end{equation}
If $\eq$ is \emph{pure} (i.e. $\supp(\eq_{\play}) = \{\peq_{\play}\}$ for some $\peq_{\play}\in\pures_{\play}$ and all $\play\in\players$), then it is called a \emph{pure equilibrium}.
In addition, $\eq$ is said to be \emph{strict} if $\eq$ is pure and
\eqref{eq:NE-pures} holds as a strict inequality for all $\pure_{\play}'\notin\supp(\eq_{\play})$, $\play\in\players$.
Equivalently, $\eq$ is a strict equilibrium if every player has a \emph{unique} best response to their opponents' strategy profile.

In generic games (i.e. games with no payoff ties), pure equilibria are also strict, so our analysis will focus throughout on strict \aclp{NE}.
%Of course, \aclp{NE} in pure strategies do not always exist (implying the same for strict equilibria as well), but they do exist in many important classes of games, ranging from the prisoner's dilemma and its $\nPlayers$-player extensions to coordination/anti-coordination games, competition games, potential/congestion games, etc.
With this in mind, we derive here a variational characterization of strict equilibria that plays a key role in our analysis:

\begin{proposition}
\label{prop:strict-var}
The profile $\eq$ is a strict equilibrium if and only if
%there exists a neighborhood $U$ of $\eq$ in $\strats$ and some real number $m>0$ such that
\begin{equation}
\label{eq:strict-var}
\braket{\payv(x)  }{x - \eq}
	\leq - \tfrac{1}{2} a \norm{x - \eq}
	\quad
%	\text{for all $x\in U$},
	\text{for some $a>0$ and for all $x$ near $\eq$},
\end{equation}
where $\norm{x} = \sum_{\play} \sum_{\pure\in\pures_{\play}} \abs{x_{\play\pure_{\play}}}$ denotes the $L^{1}$-norm of $x$.
%where $\braket{\payv(x)}{x - \eq} = \sum_{\play} \braket{\payv_{\play}(x)}{x_{\play} - \eq_{\play}}$ is the canonical  inner product  between vectors. Here, $\norm{x} = \sum_{k} \abs{x_{k}}$ denotes the $L^{1}$-norm of $x$.
\end{proposition}

Motivated by this characterization of strict equilibria (which we prove in Appendix \ref{app:auxiliary}), we say that $\eq$ is a \emph{globally strict} equilibrium if \eqref{eq:strict-var} holds for all $x\in\strats$ \textendash\ for instance, as is easily seen to be the case in the Prisoner's Dilemma.
%generic competition games, potential games with a unique equilibrium, etc.
Obviously, if $\eq$ is globally strict, then it is the unique equilibrium of the game, similarly to the notion of a \ac{GESS} in evolutionary game theory \cite{San10}.

\section{Learning via hedging}
\label{sec:HEDGE}
%----------------------------------------------------------------------
%%% LEARNING
%----------------------------------------------------------------------
% !TEX root = ./WINE.tex

The algorithm that we examine is the so-called ``HEDGE'' variant of the \acl{EW} algorithm \cite{FS99}.
In a nutshell, the main idea of the algorithm is as follows:
At each stage $t=1,2,\dotsc$ of the process, players maintain and update a ``performance score'' for each of their actions (pure strategies) based on each action's cumulative payoff up to stage $t$.
These scores are then converted to mixed strategies by assigning exponentially higher probability to actions with higher scores;
subsequently, a new action is drawn based on these mixed strategies, and the process repeats.
More precisely, we have the following iterative algorithm:
%It starts with some initial mixed strategies $x$ which each player $\play$ uses $x_{\play}(0)$ for the first round.
%After each round $t$, each player $\play$ compute a new  mixed strategy $x_{\play}(t+1)$ according its observations and $x_{\play}(t)$.
%The notation  ``$\pure\sim x$'' means that an action $\pure$ is chosen based on the distribution $x$. More formally,

\medskip

\begin{algorithm}{HEDGE with variable step-size $\step_{t}$}
\smallskip
%		\label{exp3}
%	    \\$\textbf{Input:}$Let $\eta \in \mathbb{R}^{+}$ and $\gamma$, $\beta$  in $(0,1]$  	  \\
%	    $\textbf{Initialization}$ : Let $p_{1}$ be the uniform distribution over $\{1,...,K\}$ 
 %\\  $\gamma = min \{1,\sqrt{ \frac{K ln K}{(e-1)g}}\}$, with $g\geq G_{max}$
\\
Each player $\play \in \players$ has an
initial score vector $y_{\play}$ and plays with initial mixed strategy $x_{\play} = \logit_{\play}(y_{\play})$
where the logit map $\logit_{\play}$ is defined as
\begin{equation}
\label{eq:logit}
\logit_{\play}(y_{\play})
	= \frac{1}{\sum_{\pure\in\pures_{\play}} \exp(y_{\play\pure})}
	(\exp(y_{\play\pure}))_{\pure\in\pures_{\play}}.
\end{equation}
% initial mixed strategy $x_{\play}$ and a score vector $y_{\play}$.
% such that all its elements are set to $1$.
     \\   \For{each round $t$  }
\>
	\\   Each player $\play \in \players$ draws a pure strategy  $\pure_{\play}$ according to $x_{\play}$
%                    (\emph{i.e.} $\pure_{\play}(t)\sim x_{\play}(t ))$ 
                  \\      Each player  $\play$ observes their individual payoff vector $\payv_{i}(x)$
%                  perfect knowledge of the payoff vector $\payv_{\play}(x)$ given the players' mixed strategies.
		 \\    Each player $\play$  updates their mixed strategy $x_{\play}$ via the recursion
\<
\begin{equation}
\label{eq:HEDGE}
\tag{HEDGE}
\begin{aligned}
y_{\play}
	&\leftarrow y_{\play} + \step_{t} \payv_{\play}(x),
	\\
x_{\play}
	&\leftarrow \logit_{\play}(y_{\play})
\end{aligned}
\end{equation}
%\<
\End\For
\end{algorithm}

\medskip
 
As stated above,  \eqref{eq:HEDGE} tacitly assumes that players have perfect knowledge of their \emph{mixed} payoff vectors $\payv_{\play}(x(t))$ at each iteration of the algorithm.
However, in practical applications of game theory \textendash\ especially in large networks and telecommunication systems \textendash\ this assumption is often too stringent.
For this reason, much of our analysis will concern the case where players only have access to a possibly imperfect estimate $\hat\payv(t)$ of their \emph{pure} payoff vector $\payv_{\play}(\pure(t)) \equiv (\pay_{\play}(\pure;\pure_{-\play}(t))_{\pure\in\pures_{\play}}$ given the pure strategy profile $\pure(t)\in\pures$ drawn at stage $t$.
%based on the players' mixed strategy profile $x(t)\in\strats$.
In other words, we will be interested in the case where players can only estimate the payoff of their pure strategies given the chosen actions of all other players.

Formally, this can be represented by the general feedback model
\begin{equation}
\label{eq:payv-noise}
\hat\payv_{\play}(t)
	= \payv_{\play}(\pure(t)) + \noise_{\play}(t),
\end{equation}
where the error process $\noise = (\noise_{\play})_{\play\in\players}$ satisfies the statistical hypotheses
%conditioned on the history $\filter_{t}$ of $x(t)$ up to stage $t$:%
%\footnote{Formally, $\filter_{t}$ stands here for the natural filtration of $x(t)$;
%as a result, $x(t)$ is adapted to $\filter_{t}$ but $\noise_{t}$ is not (note that the natural filtrations of $x_{t}$ and $\noise_{t}$ differ by a stage).}
\begin{enumerate}
[\indent 1.]
%[({H}1)]
\item
\emph{Zero-mean:}
\begin{alignat*}{2}
\tag{H1}
\label{eq:zeromean}
\exof*{\noise(t) \given \filter_{t}}
	&= 0
	&\quad
	\text{for all $t=1, 2,\dotsc$ \as.}
%\end{equation}
\intertext{%
\item
\emph{Finite \acl{MSE}:}
there exists some $\noisedev>0$ such that}
%\begin{equation}
\tag{H2}
\label{eq:MSE}
\txs
\exof{\dnorm{\noise(t)}^{2} \: \given \filter_{t}}
	&\leq \noisevar
	&\quad
	\text{for all $t=1, 2,\dotsc$ \as.}
%\end{equation}
\end{alignat*}
\end{enumerate}
In the above, the expectation $\exof{\argdot}$ is taken with respect to the randomness induced by the players' mixed strategies and the error process $\noise$, while $\filter_{t}$ denotes the history of $x(t)$ up to stage $t$.%
\footnote{Formally, $\filter_{t}$ is defined as the natural filtration induced by $x(t)$ \cite{HH80}.}
Put differently, Hypotheses \eqref{eq:zeromean} and \eqref{eq:MSE} simply mean that the players' estimates $\hat\payv_{\play}$ are \emph{conditionally unbiased and bounded in mean square}, i.e.
\begin{subequations}
\label{eq:estimates}
\begin{flalign}
\label{eq:unbiased}
&\exof{\hat\payv(t) \given \filter_{t}}
	= \payv(x(t)),
	\\
\label{eq:vbound}
&\exof{ \dnorm{\hat\payv(t)}^{2} \given \filter_{t} }
	\leq \vbound^{2},
%	< +\infty
%	\quad
%	\text{for some $\vbound < \infty$}.
\end{flalign}
\end{subequations}
where $\vbound>0$ is a finite positive constant (in the noiseless case, $\vbound$ is simply the players' maximum absolute payoff).
Thus, Hypotheses \eqref{eq:zeromean} and \eqref{eq:MSE} allow for a broad range of noise distributions, including all compactly supported, \mbox{(sub-)}Gaussian, (sub-)ex\-po\-nen\-tial and log-normal distributions.
%\footnote{In particular, we will not be assuming \ac{iid} errors;
%this point is crucial for applications to practical systems where measurements are often correlated with the state of the system.}

%----------------------------------------------------------------------
%%% RESULTS
%----------------------------------------------------------------------
\section{Analysis and results}
\label{sec:results}
%----------------------------------------------------------------------
%%% RESULTS
%----------------------------------------------------------------------
% !TEX root = ./WINE.tex

%----------------------------------------------------------------------
%%% RESULTS
%----------------------------------------------------------------------
\subsection{Statement of the results}
\label{sec:statements}

In this section, we provide our main convergence results for the algorithm \eqref{eq:HEDGE}.
%with both perfect and imperfect information.
%(following the general feedback model outlined in the previous section).
For simplicity, we start with the perfect information case:%
%\footnote{To maintain the flow of the paper, we postpone all proofs to the end of this section.}

\begin{theorem}
\label{thm:conv-det}
Suppose that \eqref{eq:HEDGE} is run with perfect mixed payoff observations and is initialized not too far from $\eq$.
Then, we have
\begin{equation}
\label{eq:rate-det}
\norm{x(t) - \eq}
	\leq C e^{-a \sum_{j=1}^{t}\step_{j}},
\end{equation}
where $C > 0$ is a constant that depends on the initialization of \eqref{eq:HEDGE} and $a>0$ is a constant that only depends on the game.
In particular, if $\sum_{t=1}^{\infty} \step_{t}=\infty$, $x(t)$ converges to $\eq$.
%and a step-size of the form $\step_{t}\propto1/t$, we have $\norm{x(t) - \eq} = \bigoh(1/t)$.
\end{theorem}

\begin{corollary}
\label{cor:rate-det-opt}
If the algorithm \eqref{eq:HEDGE}  is run with assumptions as above and a constant step-size $\step$, we have
\begin{equation}
\label{eq:rate-det-opt}
\norm{x(t) - \eq}
	= \bigoh(e^{-a\step t}).
\end{equation}
\end{corollary}

\begin{remark}
The locality of Theorem \ref{thm:conv-det} has to do with the fact that a game may admit several strict equilibria, so the algorithm's end state depends on its initialization.
Instead,
%if the game admits a unique, globally strict equilibrium $\eq$, \eqref{eq:HEDGE} converges to $\eq$ from any initialization.
if $\eq$ is globally strict (meaning that the game admits a unique \acl{NE}), the above results hold globally and there is no dependence on the algorithm's initialization (for a precise statement, see Theorem \ref{thm:conv-global} below).
\end{remark}

\begin{remark}
We should also note here that the convergence rate \eqref{eq:rate-det} \emph{improves} with larger step-sizes $\step_{t}$.
The reason for this (fairly surprising) behavior is that, when there are no estimation errors, the algorithm consistently reinforces the players' equilibrium strategies near a strict equilibrium.
As a result, in the absence of uncertainty, players can employ \eqref{eq:HEDGE} in a very greedy fashion and achieve arbitrarily fast convergence rates \textendash\ in stark contrast to standard results in game-theoretic learning and convex optimization which often require a small, decreasing step-size.
We show below that this property is inextricably tied to the absence of uncertainty:
if the players' observations are affected by even a modicum of randomness, it is necessary to use a more conservative step-size policy (cf. Theorem \ref{thm:conv-stoch}).
\end{remark}

\begin{remark}
Regarding the game's dimensionality (i.e. the number of players and actions per player), it can be shown that $C = \bigoh(\sum_{\play\in\players} \abs{\pures_{\play}})$ while $a$ depends only on the relative differences between the players' payoffs \textendash\ specifically, we can take $a = \min_{\play}\min_{\pure_{\play}\neq\peq_{\play}} \left[ \pay_{\play}(\peq) - \pay_{\play}(\pure_{\play};\peq_{-\play}) \right] > 0$.
In other words, the algorithm's half-life is asymptotically \emph{independent} of the size of the game.
%In other words, the algorithm's convergence rate does not depend on the number of players or pure strategies per player, but only on the game's relative payoff differences.
\end{remark}

The basic ingredient of the proof of Theorem \ref{thm:conv-det} (presented at the end of this section) is as follows.
First, assuming the algorithm starts relatively close to a given strict equilibrium, we show that the induced sequence of play always remains nearby.
Then, by studying the evolution of the players' score variables $y(t)$, we show that the cumulative payoff difference between a player's equilibrium strategy and all other pure strategies grows asymptotically as $\bigoh(\sum_{j=1}^{t} \step_{j})$ for large $t$.
%$\bigoh(\theta_{t})$, where $\theta_{t} = \sum_{j=1}^{t} \step_{j}$ is the cumulative length of the player's steps up to iteration $t$.
The derived exponential rate is then a consequence of the properties of the logit map $\logit$.

On the other hand, if the players' payoff observations are subject to noise and stochastic uncertainty, a single unlucky estimation could drive $x(t)$ away from the basin of a strict equilibrium, possibly never to return.
As a result, any local convergence result in the presence of noise must be probabilistic in nature.
This is emphasized in our next result which shows that convergence can be achieved with probability arbitrarily close to $1$:

\begin{theorem}
\label{thm:conv-stoch}
Fix a confidence level $\eps>0$ and suppose that the algorithm \eqref{eq:HEDGE} is run with
a small enough step-size $\step_{t}$ satisfying $\sum_{t=1}^{\infty} \step_{t}^{2} < \sum_{t=1}^{\infty} \step_{t} = \infty$
and
imperfect pure payoff information satisfying Hypotheses \eqref{eq:zeromean} and \eqref{eq:MSE}.
If $\eq$ is a strict equilibrium and \eqref{eq:HEDGE} is initialized not too far from $\eq$, we have
\begin{equation}
\label{eq:rate-stoch}
\probof*{\norm{x(t) - \eq} \leq C' e^{-a\sum_{j=1}^{t}\step_{j}} \, \text{for all $t$}}
	\geq 1-\eps,
\end{equation}
where
$a>0$ is a constant that only depends on the game
and
$C'>0$ is a \textpar{random} constant that depends on the initialization of \eqref{eq:HEDGE}.
In particular, under the stated assumptions, $x(t)\to\eq$ with probability at least $1-\eps$.
\end{theorem}

\begin{corollary}
\label{cor:rate-stoch-opt}
With assumptions as above, if the algorithm \eqref{eq:HEDGE} is run with a step-size of the form $\step_{t} = \step/t^{\beta}$ for some sufficiently small $\step>0$ and $\beta\in(1/2,1)$, we have
\begin{equation}
\label{eq:rate-stoch-opt}
\probof*{\norm{x(t) - \eq} = \bigoh\left(e^{-\frac{a\step}{1-\beta}t^{1-\beta}}\right)}
	\geq 1-\eps,
\end{equation}
\end{corollary}

\begin{remark}
In contrast to the full information case, the ``$\ell^{2} - \ell^{1}$'' summability requirement $\sum_{t=1}^{\infty} \step_{t}^{2} < \sum_{t=1}^{\infty} \step_{t} = \infty$ constrains the admissible step-size policies that lead to strict equilibrium (for instance, constant step-size policies are no longer admissible).
In particular, the most aggressive step-size that can be used in the presence of noise is $\step_{t} \propto t^{-\beta}$ for some $\beta$ close (but not equal) to $1/2$, leading to a convergence rate of $\lambda^{t^{1-\beta}}$ for some $\lambda<1$ (cf. Corollary \ref{cor:rate-stoch-opt}).
This bound on $\beta$ is due to the second moment control required by Doob's maximal inequality;
if there is finer control on the moments of the noise process $\noise$ (for instance, if the noise is sub-exponential), the lower bound $\beta > 1/2$ can be pushed all the way down to $\beta > 0$, implying a quasi-linear convergence rate.

\end{remark}

As was hinted above, the main idea behind the proof of Theorem \ref{thm:conv-stoch} is to use Doob's maximal inequality for martingales to show that the probablity of $x(t)$ escaping the basin of attraction of a strict equilibrium $\eq$ can be made arbitrarily small if the algorithm's step-size is chosen appropriately.
Once this probabilistic estimate is in place, convergence is obtained roughly as in the case of Theorem \ref{thm:conv-det}.

\smallskip

Building on this, if $\eq$ is \emph{globally} strict, we have the stronger result:
\vspace{-2pt}

\begin{theorem}
\label{thm:conv-global}
Suppose that the algorithm \eqref{eq:HEDGE} is run with a step-size $\step_{t}$ such that $\sum_{t=1}^{\infty} \step_{t}^{2} < \sum_{t=1}^{\infty} \step_{t} = \infty$ and imperfect pure payoff observations satisfying Hypotheses \eqref{eq:zeromean} and \eqref{eq:MSE}.
If $\eq$ is a globally strict equilibrium, then:
\vspace{-5pt}
\begin{enumerate}
\addtolength{\itemsep}{2pt}
\item
$x(t)\to\eq$ \as.
\item
There exists a \textpar{deterministic} constant $c>0$ depending only on the game such that
\begin{equation}
\label{eq:rate-global}
\norm{x(t) - \eq}
	= \bigoh(e^{-c \sum_{j=1}^{t} \step_{j}}).
\end{equation}
\end{enumerate}
\end{theorem}

\begin{corollary}
\label{cor:rate-global}
With assumptions as above, if \eqref{eq:HEDGE} is run with a step-size of the form $\step_{t} = \step/t^{\beta}$ some $\beta\in(1/2,1)$, we have
\begin{equation}
\label{eq:rate-global-opt}
\norm{x(t) - \eq}
	= \bigoh(e^{-\frac{a\step}{1-\beta} t^{1-\beta}}).
\end{equation}
\end{corollary}

As opposed to Theorems \ref{thm:conv-det} and \ref{thm:conv-stoch}, the proof of Theorem \ref{thm:conv-global} relies heavily on the so-called \acdef{KL} divergence \cite{kullback1951information}, defined here as
\begin{equation}
\label{eq:KL}
\dkl(\eq,x)
	= \sum_{\play\in\players} \sum_{\pure_{\play}\in\pures_{\play}}
	\eq_{\play\pure_{\play}} \log \frac{\eq_{\play\pure_{\play}}}{x_{\play\pure_{\play}}}
	\quad
	\text{for all $x\in\strats^{\circ}$}.
\end{equation}
%where $\eq$ is the strict equilibrium in question and $x\in\strats^{\circ}$.
The \ac{KL} divergence is a positive-definite, asymmetric distance measure that is particularly well-adapted to the analysis of the replicator dynamics \cite{Wei95,San10,LM13}.
By using this divergence as a discrete-time Lyapunov function, we show that $\eq$ is a \emph{recurrent point} of the process $x(t)$, i.e. $x(t)$ visits any neighborhood of $\eq$ infinitely many times.
We then use an argument similar to the proof of Theorem \ref{thm:conv-stoch} to show that the process actually converges to $\eq$ at an asymptotic rate of $\bigoh(e^{-c\sum_{j=1}^{t}\step_{j}})$.

The step-size assumption in the statement of Theorem \ref{thm:conv-global} is key in achieving this, but it is important to note it can be relaxed to the lighter requirement $\sum_{j=1}^{t}\step_{j}^{2} \big/ \sum_{j=1}^{t}\step_{t} \to 0$ if the players' feedback noise is bounded (for instance, if players have access to their actual pure payoff information).
When this is the case, it is possible to achieve a convergence rate of the form $\bigoh(e^{-ct^{1-\beta}})$ for any $\beta>0$ by using a step-size sequence of the form $\step_{t}\propto 1/t^{\beta}$.
Finally, we should also note that the multiplicative constant in \eqref{eq:rate-global} is $\bigoh(\sum_{\play\in\players} \abs{\pures_{\play}})$, i.e. it is linear in the dimensionality of the game (just as in the case of Theorems \ref{thm:conv-det} and \ref{thm:conv-stoch}).
As for the constant $c>0$, \eqref{eq:rate-global} holds for all $c<a$ (where $a>0$ is the payoff-based convergence rate established in Theorems \ref{thm:conv-det} and \ref{thm:conv-stoch}), showing that Theorem \ref{thm:conv-global} guarantees essentially the same exponential convergence rate as Theorems \ref{thm:conv-det} and \ref{thm:conv-stoch}.

%----------------------------------------------------------------------
%%% PROOFS
%----------------------------------------------------------------------
\subsection{Proofs}
\label{sec:proofs}

Below we provide the proofs of the above results, relegating some technical details to Appendix \ref{app:auxiliary}.
For simplicity, we begin with the perfect information case (Theorem \ref{thm:conv-det});
we then build on this analysis to prove our convergence results in the presence of uncertainty (Theorems \ref{thm:conv-stoch} and \ref{thm:conv-global}).

\paragraph{Proof of Theorem \ref{thm:conv-det}.}
Suppose that $\eq = (\peq_{1},\dotsc,\peq_{\nPlayers})$ is a strict equilibrium.
Then, by continuity, there exists some $a>0$ and a neighborhood $U$ of $\eq$ such that $\payv_{\play\peq_{\play}}(x) - \payv_{\play\pure_{\play}}(x) \geq a$ for all $x\in U$ and all $\pure_{\play}\in\pures_{\play}\setminus\{\peq_{\play}\}$, $\play\in\players$.

Now, introduce the auxiliary variables
\begin{equation}
\label{eq:zdef}
z_{\play\pure_{\play}}
	= y_{\play\pure_{\play}} - y_{\play\peq_{\play}},
\end{equation}
where $y_{\play\pure_{\play}}$ represents the cumulative payoff score of strategy $\pure_{\play}\in\pures_{\play}$ (cf. the definition of \eqref{eq:HEDGE} in the previous section).
Then, for all $M>0$, Proposition \ref{prop:Fenchel} in Appendix \ref{app:auxiliary} shows that the set $U_{M} = \setdef{x = \logit(y)}{\text{$z_{\play\pure_{\play}}\leq -M$ for all $\pure_{\play}\in\pures_{\play}$, $\play\in\players$}}$ is a neighborhood of $\eq$ in $\strats^{\circ}$ which is contained in $U$ if $M$ is chosen large enough.
Thus, if $x(t) \in U_{M}$, we get:
\begin{subequations}
\begin{flalign}
\label{eq:zbound1}
z_{\play\pure_{\play}}(t+1)
	&= z_{\play\pure_{\play}}(t)
	+ \step_{t} \left[ \payv_{\play\pure_{\play}}(x(t)) - \payv_{\play\peq_{\play}}(x(t)) \right]
	\\
\label{eq:zbound2}
	&\leq z_{\play\pure_{\play}}(t) - a \step_{t}
	\\
\label{eq:zbound3}
	&\leq -M - a\step_{t}
\end{flalign}
\end{subequations}
implying in particular that $x(t+1) \in U_{M}$ as well.
Thus, if \eqref{eq:HEDGE} is initialized in $U_{M}$, we obtain by induction that $x(t)\in U_{M} \subseteq U$ for all $t$.

To proceed, given that $x(t)\in U_{M}$ for all $t$, telescoping \eqref{eq:zbound2} yields
%\(
\begin{equation}
\label{eq:zbound4}
z_{\play\pure_{\play}}(t+1)
	\leq -M - a \sum_{j=1}^{t} \step_{j}.
\end{equation}
%\)
Hence, from the definition of $\logit$ (see \eqref{eq:logit}), we obtain:
\begin{flalign}
\label{eq:xbound1}
x_{\play\peq_{\play}}(t+1)
%	= \frac{\exp(y_{\play\peq_{\play}}(t+1))}{\sum_{\pure_{\play}\in\pures_{\play}} \exp(y_{\play\pure_{\play}}(t+1))}
	&= \frac{1}{1 + \sum_{\pure_{\play}\neq\peq_{\play}} \exp(z_{\play\pure_{\play}}(t+1))}
	\notag\\
	&\geq 1 - \sum_{\pure_{\play}\neq\peq_{\play}} \exp(z_{\play\pure_{\play}}(t+1))
	\notag\\
	&\geq 1 - \sum_{\pure_{\play}\neq\peq_{\play}} e^{-M} e^{-a\sum_{j=1}^{t}\step_{j}}.
\end{flalign}
Therefore, since $\norm{x_{\play} - \eq_{\play}} = 1 - x_{\play\peq_{\play}} + \sum_{\pure_{\play}\neq\peq_{\play}} x_{\play\pure_{\play}} = 2 (1 - x_{\play\peq_{\play}})$, rearranging \eqref{eq:xbound1} yields
\begin{equation}
\norm{x(t+1) - \eq}
	\leq 2 \sum_{\play\in\players} \sum_{\pure_{\play}\neq\peq_{\play}} e^{-M} e^{-a\sum_{j=1}^{t}\step_{j}},
\end{equation}
and our assertion follows.
\hfill
\qed

\medskip

We now turn to feedback imperfections, starting with Theorem \ref{thm:conv-stoch}:

\paragraph{Proof of Theorem \ref{thm:conv-stoch}.}
With notation as in the proof of Theorem \ref{thm:conv-det}, set $z_{\play\pure_{\play}} = y_{\play\pure_{\play}} - y_{\play\peq_{\play}}$ and let $M>0$ be such that $\pay_{\play\peq_{\play}}(x) - \pay_{\play\pure_{\play}}(x) \geq a$ for all $\pure_{\play}\in\pures_{\play}\setminus\{\peq_{\play}\}$, $\play\in\players$, whenever $x\in U_{M}$.
Then, we have
\begin{flalign}
\label{eq:zbound-stoch1}
z_{\play\pure_{\play}}(t+1)
	&= z_{\play\pure_{\play}}(t)
	+ \step_{t} \left[ \payv_{\play\pure_{\play}}(x(t)) - \payv_{\play\peq_{\play}}(x(t)) \right]
	+ \step_{t} \eta_{\play\pure_{\play}}(t),
\end{flalign}
where $\eta_{\play\pure_{\play}}(t) = \hat\payv_{\play\pure_{\play}}(t) - \payv_{\play\pure_{\play}}(x(t)) - (\hat\payv_{\play\peq_{\play}}(t) - \payv_{\play\peq_{\play}}(x(t)))$.
Thus, assuming that \eqref{eq:HEDGE} is initialized in $U_{2M}$ and telescoping, we get
\begin{equation}
\label{eq:zbound-stoch2}
z_{\play\pure_{\play}}(t+1)
	\leq -2M
	+ \sum_{j=1}^{t} \step_{j} \left[ \payv_{\play\pure_{\play}}(x(j)) - \payv_{\play\peq_{\play}}(x(j)) \right]
	+ \sum_{j=1}^{t} \step_{j} \eta_{\play\pure_{\play}}(j).
\end{equation}

We now claim that, if $\step_{t}$ is chosen appropriately, we have
\begin{equation}
\probof*{\sup_{t} \insum_{j=1}^{t} \step_{j} \eta_{\play\pure_{\play}}(j) \leq M} \geq 1 - \eps / (\nPlayers (S_{\play} - 1)),
\end{equation}
where $S_{\play} = \abs{\pures_{\play}}$.
Indeed, let $X_{\play\pure_{\play}}(t) = \sum_{j=1}^{t}\step_{t} \eta_{\play\pure_{\play}}(t)$ and let $E_{\play}(t)$ denote the event $\sup_{1\leq j\leq t} \abs{X_{\play\pure_{\play}}(j)} \geq M$.
By Hypothesis \eqref{eq:zeromean}, $X_{\play\pure_{\play}}(t)$ is a martingale so Doob's maximal inequality \cite[Theorem 2.1]{HH80} yields
\begin{equation}
\label{eq:Doob}
\probof{E_{\play}(t)}
	\leq \frac{\exof{X_{\play\pure_{\play}}(t)^{2}}}{M^{2}}
	\leq \frac{2\noisevar \sum_{j=1}^{t} \step_{j}^{2}}{M^{2}},
\end{equation}
where we used the noise variance estimate
\begin{equation}
\exof{\eta_{\play\pure_{\play}}^{2}(t)}
	= \exof{\exof{\eta_{\play\pure_{\play}}^{2}(t) \given \filter_{t}}}
	\leq 2 \exof{ \exof{\dnorm{\noise_{\play}(t)}^{2} \given \filter_{t}} }
	\leq 2 \noisevar,
\end{equation}
and the fact that $\exof{\eta_{\play\pure_{\play}}(t) \eta_{\play\pure_{\play}}(t')} = 0$ if $t\neq t'$.
Since $E_{\play}(t+1) \subseteq E_{\play}(t) \subseteq\dotsc$, it follows that the event $E_{\play} = \intersect_{t=1}^{\infty} E_{\play}(t)$ occurs with probability $\probof{E_{\play}} \leq 2\noisevar\Gamma_{2}/M^{2}$ where $\Gamma_{2} = \sum_{t=1}^{\infty} \step_{t}^{2} < \infty$.
Thus, if $\step_{t}$ is chosen so that $\Gamma_{2} \leq \eps M^{2}/(2\nPlayers(S_{\play} - 1)\noisevar)$, we get $\probof{\text{$X_{\play\pure_{\play}}(t) \geq M$ for all $t$}} \leq \eps/(\nPlayers(S_{\play} - 1))$.

Assume therefore that $\Gamma_{2} \leq \eps M^{2}/(2\nPlayers(S_{\play} - 1)\noisevar)$.
Then, we obtain
\begin{equation}
\label{eq:pbound}
\probof*{\max\nolimits_{\play\in\play,\pure_{\play}\in\pures_{\play}}\sup\nolimits_{t}
	X_{\play\pure_{\play}}(t) \geq M}
%	\text{for some $\pure_{\play}\in\pures_{\play}$, $\play\in\players$}}
%	\notag\\
	\leq \sum_{\play\in\players} \sum_{\pure_{\play}\neq\peq_{\play}} \frac{\eps}{\nPlayers(S_{\play} - 1)}
	\leq \eps.
\end{equation}
Hence, going back to \eqref{eq:zbound-stoch2}, a straightforward induction shows that $x(t)\in U_{M}$ for all $t$ with probability at least $1-\eps$.
When this occurs, we also have
%\(
\begin{equation}
\probof{z_{\play\pure_{\play}}(t+1) \leq -M + a\sum_{j=1}^{t}\step_{j} \, \text{for all $n$}}
	\leq 1-\eps,
\end{equation}
%\)
and the bound \eqref{eq:rate-stoch} is obtained as in the proof of Theorem \ref{thm:conv-stoch}.
\hfill
\qed

%\paragraph{Proof of Corollary \ref{cor:rate-stoch-opt}.}
%Simply note that, for all $\beta\in(1/2,1)$, $\sum_{j=1}^{t-1} t^{-\beta} \geq \int_{1}^{t} s^{-\beta} \dd s = \frac{1}{1-\beta} (t^{1-\beta} - 1)$.
%\hfill
%\qed

\paragraph{Proof of Theorem \ref{thm:conv-global}.}
Since $\sum_{t=1}^{\infty} \step_{t} = \infty$, it clearly suffices to prove \eqref{eq:rate-global}.
%To that end, we rely heavily on the so-called \acdef{KL} divergence \cite{kullback1951information}, defined here as
%\begin{equation}
%\label{eq:KL}
%\dkl(\eq,x)
%	= \sum_{\play\in\players} \sum_{\pure_{\play}\in\pures_{\play}}
%	\eq_{\play\pure_{\play}} \log \frac{\eq_{\play\pure_{\play}}}{x_{\play\pure_{\play}}},
%\end{equation}
%where $\eq$ is the strict equilibrium in question and $x\in\strats^{\circ}$.
Then, given that $\eq = (\peq_{1},\dotsc,\peq_{\nPlayers})$ is pure,
%we readily get $\dkl(\eq,x) = - \sum_{\play\in\players} \log x_{\play\peq_{\play}}$
an easy calculation yields
\begin{flalign}
\label{eq:Dgrowth}
\dkl(\eq,x)
	&= -\sum_{\play\in\players} \log x_{\play\peq_{\play}}
	= - \sum_{\play\in\players} \log( 1-(1-x_{\play\pure_{\play}^{\ast}}))
	\notag\\
	&= - \sum_{\play\in\players} \log(1 - \norm{x_{\play} - \eq_{\play}}/2)\notag\\
	&\geq \frac{1}{2} \norm{x - \eq},
\end{flalign}
so it suffices to show that $\dkl(\eq,x(t))\to0$.
%\footnote{Recall here that $\norm{x_{\play} - \eq_{\play}} = (1 - x_{\play\peq_{\play}}) + \sum_{\pure_{\play}\neq\peq_{\play}} x_{\play\pure_{\play}} = 2(1 - x_{\play\peq_{\play}})$.}
With this in mind, let $D_{t} = \dkl(\eq,x(t))$.
Then, Proposition \ref{prop:Fenchel} in Appendix \ref{app:auxiliary} yields
\begin{flalign}
\label{eq:Dbound-stoch1}
D_{t+1}
%	&\leq D_{t} + \step_{t} \braket{\hat\payv(t)}{x(t) - \eq} + \step_{t}^{2} \vbound^{2}
%	\notag\\
	&\leq D_{t}
	+ \step_{t} \braket{\payv(x(t))}{x(t) - \eq}
	+ \step_{t} \snoise_{t}
	+ \frac{1}{2}\step_{t}^{2} \dnorm{\hat\payv(t)}^{2},
\end{flalign}
where we have set $\snoise_{t} = \braket{\noise(t)}{x(t) - \eq}$.
Using this bound, we will show that $x(t)$ visits any neighborhood $U$ of $\eq$ infinitely many times.

Indeed, assume on the contrary that this is not so.
Then, by Proposition \ref{prop:strict-var}, there exists some $\delta>0$ such that $\braket{\payv(x(t))}{x(t) - \eq} \leq -\alpha\delta$ for all sufficiently large $t$.
Hence, telescoping \eqref{eq:Dbound-stoch1} yields
\begin{flalign}
\label{eq:Dbound-stoch2}
D_{t+1}
	&\leq D_{0}
	- a\delta \sum_{j=1}^{t} \step_{j}
	+ \sum_{j=1}^{t} \step_{j} \snoise_{j}
	+ \frac{1}{2} \sum_{j=1}^{t} \step_{j}^{2} \dnorm{\hat\payv(j)}^{2}
	\notag\\
	&\leq D_{0}
	- \theta_{t} \left[
		a\delta
		- \frac{\sum_{j=1}^{t} \step_{j}\snoise_{j}}{\theta_{t}}
		- \frac{\sum_{j=1}^{t} \step_{j}^{2} \dnorm{\hat\payv(j)}^{2}}{2\theta_{t}}.
		\right]
\end{flalign}
Since $\exof{\snoise_{j} \given \filter_{j}} = \exof{\braket{\noise_{j}}{x(j) - \eq} \given \filter_{j}} = \braket{\exof{\noise_{j} \given \filter_{j}}}{x(j) - \eq} = 0$ and $\exof{\dnorm{\snoise_{j}}^{2}} \leq 4 \exof{ \dnorm{\noise_{j}}^{2} \given \filter_{j}} \leq 4 \vbound^{2}$, it follows that the sum $R_{t} = \sum_{j=1}^{t} \step_{j} \snoise_{j}$ is an $L^{2}$-bounded martinagle \cite{HH80}.

Hence, by the law of large numbers for martingale differences \cite[Theorem 2.18]{HH80}, it follows that $\theta_{t}^{-1} \sum_{j=1}^{t} \step_{j} \snoise_{j} \to 0$ \as.
Likewise, 
if we let $S_{t} = \sum_{j=1}^{t} \step_{j}^{2} \dnorm{\hat\payv(j)}^{2}$, we readily get 
\[\exof{S_{t}} = \exof{\exof{S_{t}} \given \filter_{t}} \leq \vbound^{2} \sum_{j=1}^{t} \step_{j}^{2} \leq \Gamma_{2} \vbound^{2},\]
 where $\Gamma_{2} = \sum_{j=1}^{\infty} \step_{j}^{2}$.
Hence, by Doob's martingale convergence theorem \cite[Theorem 2.5]{HH80}, $S_{t}$ converges to some (random) finite value \as.
Combining the above, we conclude that the term in the brackets of \eqref{eq:Dbound-stoch2} converges to $a\delta$ \as.
In turn, this implies that $D_{t+1} \to -\infty$, and this yields to a contradiction.

We have thus shown that $x(t)$ visits infinitely many times every neighborhood $U$ of $\eq$ \textendash\ and hence, in particular, the neighborhood $U_{2M}$ defined in the proof of Theorem \ref{thm:conv-stoch}.
Since $x(t)$ remains in $U_{2M}$ with positive probability, it follows that the probability that $x(t)$ exits $U_{2M}$ infinitely many times is zero.
We thus get $x(t) \in U_{2M}$ for all $t$ greater than some random (but finite) $t_{0}$;
hence, telescoping \eqref{eq:zbound-stoch1} starting at $t_{0}$ yields
\begin{equation}
\label{eq:zbound-stoch3}
z_{\play\pure_{\play}}(t+1)
	\leq -2M
	- \theta_{t} \left[ a - \theta_{t}^{-1} \sum\nolimits_{j=t_{0}}^{t} \step_{j} \eta_{\play\pure_{\play}}(j) \right].
\end{equation}
As above, the law of large numbers \cite[Theorem 2.18]{HH80} shows that the term in the brackets of \eqref{eq:zbound-stoch3} converges to $a$.
Hence, by \eqref{eq:xbound1}, we finally obtain $\norm{x(t) - \eq} = \bigoh(e^{-c\sum_{j=1}^{t}\step_{j}})$ for all $c\in(0,a)$, yielding our assertion.
\hfill
\qed

%----------------------------------------------------------------------
%%% PROOFS 1
%----------------------------------------------------------------------
%\section{Proofs 1}
%\label{sec:proofs-1}
%\input{Proofs-1}

%----------------------------------------------------------------------
%%% PROOFS 2
%----------------------------------------------------------------------
%\section{Proofs 2}
%\label{sec:proofs-2}
%\input{Proofs-2}

%----------------------------------------------------------------------
%%% CONCLUSIONS
%----------------------------------------------------------------------
\section{Conclusions}
\label{sec:conclusion}
%----------------------------------------------------------------------
%%% CONCLUSIONS
%----------------------------------------------------------------------
% !TEX root = ./WINE.tex

Our main goal in this paper was to analyse the convergence properties of the ``hedging'' variant of the \acl{EW} algorithm \cite{FS99} in generic $\nPlayers$-player games that admit a \acl{NE} in pure strategies.
Motivated by the applications of game theory to data networks, we focused on two different models regarding the information available to the players:
\begin{inparaenum}
[\itshape i\upshape)]
\item
perfect mixed payoff observations (where players know the payoff vector associated to their mixed strategies);
and
\item
imperfect, pure payoff observations (where players only know the payoff associated to each of their pure strategies, possibly up to some random estimation error).
\end{inparaenum}
Using the theory of stochastic approximation and discrete-time martingale processes, we show that the algorithm converges locally to a strict \acl{NE} and this convergence is exponentially fast \textendash\ even in the presence of uncertainty and noise of arbitrary magnitude.
%In such systems, we prove that our algorithm converges to a strict Nash equilibrium if initialized not too far from its equilibrium. Moreover, we provide a convergence rate.

An important extension of this work would be to consider the so-called ``bandit feedback'' setting where players are only able to observe the payoff of the action that they actually played and can only estimate the payoff of their other actions via the game's history.
Another important issue is that of asynchronicity, namely when players update at different times and there is a delay between playing and receiving feedback.
We intend to explore these directions in future work.
%Futur work includes to extend these results by reducing the information players may have on the system: for example how  the \eqref{eq:HEDGE}  algorithm can be extended in the context that  player observes only the payoff of the just played strategy?

%*************************************************************
%*****    APPENDIX
%*************************************************************
\appendix
\section{Auxiliary results}
\label{app:auxiliary}
%----------------------------------------------------------------------
%%% APP: AUXILIARY
%----------------------------------------------------------------------
% !TEX root = ./WINE.tex

\setcounter{equation}{0}
\numberwithin{equation}{section}
\numberwithin{theorem}{section}

We begin with the variational characterization \eqref{eq:strict-var} of strict \aclp{NE}:

\paragraph{Proof of Proposition \ref{prop:strict-var}.}
Assume first that $\eq$ is a strict equilibrium.
Then, for all $\play\in\players$, we have
\begin{flalign}
\label{eq:f1}
\braket{\payv_{\play}(x)}{x_{\play}-\eq_{\play}}
%	&= \sum_{\mathclap{\pure_{\play}\in\pures_{\play}}} (x_{\play\pure_{\play}} - \eq_{\play\pure_{\play}}) \, \pay_{\play}(\pure_{\play};x_{-\play}) 
%	\notag\\
	&= \pay_{\play}(x_{\play};x_{-\play}) - \pay_{\play}(\eq_{\play};x_{-\play})
	\notag\\
	%= \pay_{\play}(x) - \pay_{\play}(\peq_{\play};x_{-\play})
	&= \sum_{\pure_{\play}\neq\peq_{\play}} x_{\play\pure_{\play}} \pay_{\play}(\pure_{\play};x_{-\play})
	+ x_{\play\peq_{\play}} \, \pay_{\play}(\peq_{\play};x_{-i})
	- \pay_{\play}(\peq_{\play};x_{-\play})
	\notag\\
	& = \sum_{\pure_{\play}\neq\peq_{\play}} x_{\play\pure_{\play}}
	\left[
	\pay_{\play}(\pure_{\play};x_{-\play}) - \pay_{\play}(\peq_{\play};x_{-\play})
	\right],
\end{flalign}
where the first line is a consequence of \eqref{eq:braket} while the last one follows by noting that $\sum_{\pure_{\play}\neq\peq_{\play}} x_{\play\pure_{\play}} = 1 - x_{\play\peq_{\play}}$ and rearranging.
%The first equality is exactly the definition of the inner product whereas the second is obtained by definition of $\pay_{\play}$.
%This equation $\sum_{\pure_{\play}\neq\peq_{\play}} x_{\play\pure_{\play}} = 1 - x_{\play\peq_{\play}}$ allows to deduce  the latter equality after rearranging.

Now, by continuity \textendash\ and the fact that $\eq = (\peq_{1},\dotsc,\peq_{\nPlayers})$ is a strict equilibrium \textendash\ there exists a real number $a>0$ and a neighborhood $U$ of $\eq$ in $\strats$ such that  for all $\pure_{\play}\in\pures_{\play}\setminus\{\peq_{\play}\}$, $\pay_{\play}(\peq_{\play};x_{-\play}) - \pay_{\play}(\pure_{\play};x_{-\play}) \geq a > 0$.
Therefore:
\begin{equation}
\label{eq:f2}
\braket{\payv_{\play}(x)}{x_{\play}-\eq_{\play}}
	\leq -a \sum_{\pure_{\play}\neq\peq_{\play}} x_{\play\pure_{\play}}
\end{equation}
Hence, combining Eqs. \eqref{eq:f1} and \eqref{eq:f2}, we get the bound
\begin{equation}
\braket{\payv(x)}{x - \eq} =   \sum_{\play\in \players} \braket{\payv_{\play}(x)}{x_{\play} - \eq_{\play}}
	\leq -a \sum_{\play\in \players} \sum_{\pure_{\play}\neq\peq_{\play}} x_{\play\pure_{\play}} 
	\leq -\frac{a}{2}\norm{x_{\play} - \eq_{\play}},
\end{equation}
where the last inequality  follows from the fact that  $\eq_{\play\pure_{\play}} = 0$ if $\pure_{\play}\neq\peq_{\play}$
so $\norm{x_{\play} - \eq_{\play}} = 1 - x_{\play\peq_{\play}} + \sum_{\pure_{\play}\neq\peq_{\play}} x_{\play\pure_{\play}} = 2 \sum_{\pure_{\play}\neq\peq_{\play}} x_{\play\pure_{\play}}$.

Assume now that $\eq$ satisfies \eqref{eq:strict-var} but is not a strict \acl{NE}, so $\payv_{\play\pure_{\play}}(\eq) \leq \payv_{\play\pure_{\play}'}(\eq)$ for some $\pure_{\play}\in\supp(\eq_{\play})$, $\pure_{\play}'\in\pures_{\play}\setminus\{\pure_{\play}\}$, $\play\in\players$.
Then, if we take $x_{\play} = \eq_{\play} + \lambda (\bvec_{\play\pure_{\play}'} - \bvec_{\play\pure_{\play}})$ 
%\comment{$x_{\play} = ( \eq_{\play 1} , \dotsc , \eq_{\play \pure_{\play}'} + \lambda , \dotsc , \eq_{\play\pure_{\play}} - \lambda , \dotsc , \eq_{\play|\pures_{\play}|}) $ }
and $x_{-\play} = \eq_{-\play}$ with $\lambda>0$ small enough, we get
\begin{equation}
\braket{\payv(x)}{x - \eq}
	= \braket{\payv_{\play}(x)}{x_{\play} - \eq_{\play}}
	= \lambda \payv_{\play\pure_{\play}'}(\eq) - \lambda \payv_{\play\pure_{\play}}(\eq)
	\geq 0,
\end{equation}
in contradiction to \eqref{eq:strict-var} which yields $\braket{\payv(x)}{x-\eq} < 0$
This shows that $\eq$ is strict and completes our proof.
\hfill
$\qed$

We now prove two important properties of the logit map and the \ac{KL} divergence:
%that played a crucial role in the analysis of Section \ref{sec:proofs}:

\begin{proposition}
\label{prop:Fenchel}
Let $\pures = \{1,\dotsc,n\}$ be a finite set and let $\simplex\equiv\simplex(\pures)$ denote the $(n-1)$-dimensional simplex spanned by $\pures$.
Then:
\begin{enumerate}
\item
If $\eq\in\simplex$ is pure \textpar{i.e. $\supp(\eq) = \{\peq\}$ for some $\peq\in\pures$}, the set $U_{M} = \setdef{x = \logit(y)}{y_{\pure} - y_{\peq} \leq -M\:\text{for all $\pure\neq\peq$}}$ is a neighborhood of $\eq$ in $\intsimplex$;
furthermore, if $M$ is sufficiently large, $U_{M}$ is contained in a $\norm{\cdot}$-ball centered at $\eq$.
\item
Let $x=\logit(y)$, $x' = \logit(y')$ for some $y,y'\in\R^{n}$.
Then, we have
\begin{equation}
\dkl(\eq,x')
	\leq \dkl(\eq,x)
	+ \braket{y'-y}{x - \eq}
	+ \frac{1}{2} \dnorm{y' - y}^{2}.
\end{equation}
\end{enumerate}
\end{proposition}

\begin{proof}
For our first claim, we  assume to the contrary that $U_{M}$ is not a neighborhood of $\eq$ in $\strats^{\circ}$. So there exists a sequence $x_{k} = \logit(y_{k})$ in $\strats^{\circ}$ that converges to $\eq$, but $x_{k}\notin U_{M}$ for all $k$.
By passing to a subsequence if necessary, there exists some $\pure\in\pures$ such that $y_{k,\pure} - y_{k,\peq} > -M$ and $y_{k,\pure} \geq y_{k,\purealt}$ for all $\purealt$ (simply pick a constant subsequence of $\argmax\setdef{y_{k,\purealt}}{y_{k,\purealt} - y_{k,\peq} > -M}$ if needed).
We then get
\begin{flalign}
x_{k,\pure}
	&= \frac{e^{y_{k,\pure}}}{\sum_{\purealt} e^{y_{k,\purealt}}}
	= \frac{1}{e^{y_{k,\peq} - y_{k,\pure}} + \sum_{\purealt\neq\peq} e^{y_{k,\purealt} - y_{k,\pure}}}
 	\notag\\
	&\geq \frac{1}{e^{M}+ n -1},
\end{flalign}
contradicting the original assumption $x_{k,\pure}\to 0$ (since $x_{k}\to\eq$).

For the converse implication (namely that $U_{M}$ is contained in a ball centered at $\eq$), fix some $\delta>0$ and let $z_{\pure} = y_{\pure} - y_{\peq}$, $\pure\in\pures\setminus\{\peq\}$.
Then, letting $x=\logit(y)$ for some $y\in U_{M}$, we have
\begin{equation}
x_{\peq}
	= \frac{e^{y_{\peq}}}{\sum_{\pure\in\pures} e^{y_{\pure}}}
	= \frac{1}{1 + \sum_{\pure\neq\peq} e^{z_{\pure}}}
	\geq 1 - \sum_{\pure\neq\peq} e^{z_{\pure}}
	\geq 1 - (n - 1) e^{-M}.
\end{equation}
Thus, for $M > \frac{\abs{\log\delta}}{2(n-1)}$, we obtain
\begin{equation}
\norm{x - \eq}
	= 2(1 - x_{\peq})
	\leq 2(n-1) e^{-M}
	\leq \delta,
\end{equation}
implying that $U_{M}$ is contained in the ball $B_{\delta} = \setdef{x}{\norm{x-\eq} \leq \delta}$.

Finally, for our second claim, let $h(x) = \sum_{\pure} x_{\pure} \log x_{\pure}$, $x\in\simplex$, and let $h^{\ast}(y) = \max\{\braket{y}{x} - h(x)\} = \log\left( \sum_{\pure} e^{y_{\pure}}\right)$ denote the convex conjugate of $h$ \cite{Roc70,SS11}.
Then, a straightforward derivation yields
\begin{equation}
\label{eq:mirror}
\frac{\pd h^{\ast}}{\pd y_{\pure}}
	= \frac{\exp(y_{\pure})}{\sum_{\purealt} \exp(y_{\purealt})}
	= \logit_{\pure}(y),
\end{equation}
so, by the properties of the Legendre transform \cite{Roc70}, we get $\logit(y) = \argmax\{\braket{y}{x} - h(x)\}$.
Therefore, taking $x=\logit(y)$, the \ac{KL} divergence becomes
\begin{flalign}
\dkl(\eq,x)
	&= h(\eq) - h(x) - \braket{\nabla h(x)}{x - \eq}
	\notag\\
	&= h(\eq) + \braket{y}{x} - h(x) - \braket{y}{x} - \braket{\nabla h(x)}{x-\eq}
	\notag\\
	&= h(\eq) + h^{\ast}(y) - \braket{y}{\eq}
	\notag\\
	&\eqqcolon \fench(\eq,y),
\end{flalign}
where $\fench(\eq,y)$ is the so-called Fenchel coupling \cite{MS16}, and we used the fact that $y=\nabla h(x)$ (recall that $x = \logit(y)$ is defined as the maximizer of the quantity $\braket{y}{x} - h(x)$).

With this in mind, it suffices to show that
\begin{equation}
\label{eq:Fbound}
\fench(\eq,y')
	\leq \fench(\eq,y)
	+ \braket{y'-y}{\logit(y) - \eq}
	+ \frac{1}{2}\dnorm{y'-y}^{2}.
\end{equation}
However, since $h$ is $1$-strongly convex with respect to the $L^{1}$ norm \cite[p.~135]{SS11}, it follows that its convex conjugate $h^{\ast}$ is $1$-strongly smooth with respect to the $L^{\infty}$ norm (the dual of the $L^{1}$-norm) \cite[p.~148]{SS11}.
Specifically, this implies that
\begin{equation}
\label{eq:strong}
h^{\ast}(y')
	\geq h^{\ast}(y)
	+ \braket{y'-y}{\nabla h^{\ast}(y)}
	+ \frac{1}{2} \dnorm{y' - y}^{2}
\end{equation}
Eq.~\eqref{eq:Fbound} then follows by writing out the definition of $\fench(\eq,y')$ and then using \eqref{eq:strong} and \eqref{eq:mirror}.
\hfill
\qed
\end{proof}

%*************************************************************
%*****    BIBLIOGRAPHY
%*************************************************************
\bibliographystyle{siam}
\bibliography{wine}

\begin{thebibliography}{10}

\bibitem{BHL08}
{\sc A.~Blum, M.~T. Hajiaghayi, K.~Ligett, and A.~Roth}, {\em Regret
  minimization and the price of total anarchy}, in STOC '08: Proceedings of the
  40th annual ACM symposium on the Theory of Computing, ACM, 2008,
  pp.~373--382.

\bibitem{BM07a}
{\sc A.~Blum and Y.~Mansour}, {\em Learning, regret minimization, and
  equilibria}, in Algorithmic Game Theory, N.~Nisan, T.~Roughgarden, E.~Tardos,
  and V.~V. Vazirani, eds., Cambridge University Press, 2007, ch.~4.

\bibitem{CGM15}
{\sc P.~Coucheney, B.~Gaujal, and P.~Mertikopoulos}, {\em Penalty-regulated
  dynamics and robust learning procedures in games}, Mathematics of Operations
  Research, 40 (2015), pp.~611--633.

\bibitem{FV97}
{\sc D.~Foster and R.~V. Vohra}, {\em Calibrated learning and correlated
  equilibrium}, Games and Economic Behavior, 21 (1997), pp.~40--55.

\bibitem{FS99}
{\sc Y.~Freund and R.~E. Schapire}, {\em Adaptive game playing using
  multiplicative weights}, Games and Economic Behavior, 29 (1999), pp.~79--103.

\bibitem{HH80}
{\sc P.~Hall and C.~C. Heyde}, {\em Martingale Limit Theory and Its
  Application}, Probability and Mathematical Statistics, Academic Press, New
  York, 1980.

\bibitem{kalai2005efficient}
{\sc A.~Kalai and S.~Vempala}, {\em Efficient algorithms for online decision
  problems}, Journal of Computer and System Sciences, 71 (2005), pp.~291--307.

\bibitem{kleinberg2011load}
{\sc R.~Kleinberg, G.~Piliouras, and {\'E}.~Tardos}, {\em Load balancing
  without regret in the bulletin board model}, Distributed Computing, 24
  (2011), pp.~21--29.

\bibitem{krichene2014learning}
{\sc W.~Krichene, B.~Drigh{\`e}s, and A.~M. Bayen}, {\em Learning nash
  equilibria in congestion games}, arXiv preprint arXiv:1408.0017,  (2014).

\bibitem{kullback1951information}
{\sc S.~Kullback and R.~A. Leibler}, {\em On information and sufficiency}, The
  annals of mathematical statistics, 22 (1951), pp.~79--86.

\bibitem{LM13}
{\sc R.~Laraki and P.~Mertikopoulos}, {\em Higher order game dynamics}, Journal
  of Economic Theory, 148 (2013), pp.~2666--2695.

\bibitem{LT10}
{\sc S.~Lasaulce and H.~Tembine}, {\em Game Theory and Learning for Wireless
  Networks: Fundamentals and Applications}, Academic Press, Elsevier, 2010.

\bibitem{LW94}
{\sc N.~Littlestone and M.~K. Warmuth}, {\em The weighted majority algorithm},
  Information and Computation, 108 (1994), pp.~212--261.

\bibitem{MS16}
{\sc P.~Mertikopoulos and W.~H. Sandholm}, {\em Learning in games via
  reinforcement and regularization}, Mathematics of Operations Research,
  (2016).

\bibitem{NRTV07}
{\sc N.~Nisan, T.~Roughgarden, E.~Tardos, and V.~V. Vazirani}, eds., {\em
  Algorithmic Game Theory}, Cambridge University Press, 2007.

\bibitem{Roc70}
{\sc R.~T. Rockafellar}, {\em Convex Analysis}, Princeton University Press,
  Princeton, NJ, 1970.

\bibitem{roughgarden2015intrinsic}
{\sc T.~Roughgarden}, {\em Intrinsic robustness of the price of anarchy},
  Journal of the ACM (JACM), 62 (2015), p.~32.

\bibitem{San10}
{\sc W.~H. Sandholm}, {\em Population Games and Evolutionary Dynamics},
  Economic learning and social evolution, MIT Press, Cambridge, MA, 2010.

\bibitem{SS11}
{\sc S.~Shalev-Shwartz}, {\em Online learning and online convex optimization},
  Foundations and Trends in Machine Learning, 4 (2011), pp.~107--194.

\bibitem{NIPS2015}
{\sc V.~Syrgkanis, A.~Agarwal, H.~Luo, and R.~E. Schapire}, {\em Fast
  convergence of regularized learning in games}, in Advances in Neural
  Information Processing Systems, 2015, pp.~2989--2997.

\bibitem{VZ13}
{\sc Y.~Viossat and A.~Zapechelnyuk}, {\em No-regret dynamics and fictitious
  play}, Journal of Economic Theory, 148 (2013), pp.~825--842.

\bibitem{Vov90}
{\sc V.~G. Vovk}, {\em Aggregating strategies}, in COLT '90: Proceedings of the
  3rd Workshop on Computational Learning Theory, 1990, pp.~371--383.

\bibitem{Wei95}
{\sc J.~W. Weibull}, {\em Evolutionary Game Theory}, MIT Press, Cambridge, MA,
  1995.

\end{thebibliography}

\end{document}